\pdfoutput = 1
\documentclass[10pt, a4paper,  oneside, headinclude,footinclude, BCOR5mm]{scrartcl}

\usepackage[
nochapters, 
beramono, 
eulermath,
pdfspacing, 
dottedtoc 
]{classicthesis} 

\usepackage{bussproofs} 

\usepackage{arsclassica} 

\usepackage[T1]{fontenc} 

\usepackage[utf8]{inputenc} 

\usepackage{graphicx} 
\graphicspath{{Figures/}} 

\usepackage{enumitem} 

\usepackage{subfig} 

\usepackage{amsmath,amssymb,amsthm} 

\usepackage{varioref} 

\newenvironment{bprooftree}
  {\leavevmode\hbox\bgroup}
  {\DisplayProof\egroup}

\def\SPSB#1#2{\rlap{\textsuperscript{{#1}}}\SB{#2}}
\def\SP#1{\textsuperscript{{#1}}}
\def\SB#1{\textsubscript{{#1}}}
\def\IT#1{\textit{{#1}}}
\def \BF #1{\textbf{{#1}}}

\newtheorem{prop}{Proposition}[section]
\newtheorem{mydef}{Definition}[section]
\newtheorem{notation}{Notation}[section]
\newtheorem{example}{Example}[section]

\title{\normalfont\spacedallcaps{Sequences of Rewrites: A Categorical Interpretation}} 
\author{
   Arthur Ramos\SP{*}\\
  \texttt{afr@cin.ufpe.br}
  \and
  Ruy J. G. B. de Queiroz\SP{*}\\
  \texttt{ruy@cin.ufpe.br}
\and
  Anjolina G. de Oliveira\SP{*}\\
  \texttt{ago@cin.ufpe.br}
}

\date{} 

\begin{document}

\maketitle 


{\let\thefootnote\relax\footnotetext{* \textit{Centro de Informática, Universidade Federal de Pernambuco, Recife-PE, Brazil}}}


\renewcommand{\sectionmark}[1]{\markright{\spacedlowsmallcaps{#1}}} 
\lehead{\mbox{\llap{\small\thepage\kern1em\color{halfgray} \vline}\color{halfgray}\hspace{0.5em}\rightmark\hfil}} 

\pagestyle{scrheadings} 


\begin{abstract} 

In Martin-L\"of's Intensional Type Theory, identity type is a heavily used and studied concept. The reason for that is the fact that it is responsible for the recently discovered connection between Type Theory and Homotopy Theory. The main problem with identity types, as originally formulated, is that they are complex to understand and use. Using that fact as motivation, a much simpler formulation for the identity type was initially proposed by de Queiroz \& Gabbay (1994) \cite{Ruy4} and further developed by de Queiroz \& de Oliveira (2013) \cite{Ruy1}. In this formulation, an element of an identity type is seen as a sequence of rewrites (or computational paths). Together with the logical rules of this new entity, there exists a system of reduction rules between sequence of rewrites called \textit{LND\textsubscript{EQS}-RWS}. This system is constructed using the labelled natural deduction (i.e.\ Prawitz' Natural Deduction plus \emph{derivations-as-terms}) and is responsible for establishing how a sequence of rewrites can be rewritten, resulting in a new sequence of rewrites. In this context, we propose a categorical interpretation for this new entity, using the types as objects and the rules of rewrites as morphisms. Moreover, we show that our interpretation is in accordance with some known results, like that types have a groupoidal structure. We also interpret more complicated structures, like the one formed by a rewrite of a sequence of rewrites.

\bigskip

\end{abstract}


\section{Introduction}

Since the discovery of Univalent Models in 2005 by Vladimir Voevodsky \cite{Vlad1}, Homotopy Type Theory became an area of intensive research. One of the most important results that arose from that research was the connection between Martin-Löf's Intensional Type Theory and Homotopy Theory \cite{Steve1}. The key concept for such a connection is that of an identity type and, for this reason, it is a heavily used and studied concept of that theory. Despite having a clear  homotopical interpretation, the intensional identity type, which was originally defined by Martin-Lôf in \cite{Martin1}, has some rules that are difficult to understand and use. In particular, this is the case for the \textit{Id\textsuperscript{int}-Elimination} rule. In the words of \cite{Ruy1}, " \textit{Id\textsuperscript{int}-Elimination} is too heavily loaded with (perhaps unnecessary) information". A simplification of this rule is a desirable result, since the principle of path induction, which is used as a tool in a great amount of proofs, is a direct consequence of that rule \cite{hott}. To solve this problem, a new entity in Type Theory was proposed by \cite{Ruy1}, called the sequence of rewrites. The objective was to define the identity type using these sequence of rewrites, creating an elimination rule that still captures the intensionality of the identity type and that has much less information than the  \textit{Id\textsuperscript{int}-Elimination} rule. The following statements summarize this idea \cite{Ruy1}: 
\begin{quotation}\noindent
"It seems that an alternative formulation of propositional equality within the functional interpretation, which will be a little more elaborate than the extensional \textit{Id\SPSB{ext}{A}-type}, and simpler than the intensional \textit{Id\SPSB{int}{A}-type}, could prove more convenient from the point of view of the 'logical interpretation' ... If, on the one hand, there is an overexplicitation of information in  \textit{Id\SPSB{int}{A}}, on the other hand, in  \textit{Id\SPSB{ext}{A}} we have a case of underexplicitation. With the formulation of a proof theory for equality via labelled natural deduction we wish to find a middle ground solution between those two extremes."
\end{quotation}

\noindent Although important in formalizing this new entity and its many properties using proof theory, a denotational semantics for it was not given in \cite{Ruy1}. For this reason, the main objective of this article is to give a mathematical interpretation of a sequence of rewrites. This interpretation will use the tools of Category Theory. Nonetheless, it is impossible to understand the mathematical interpretation without understanding the computational aspects of the sequence of rewrites. Therefore, we need to make a brief effort to understand the definition of it and some of the essential proof theoretic rules. 

 

\section{Sequence of Rewrites}

The concept of the sequence of rewrites arises from the computational interpretation of the identity type. In type theory, a witness \IT{p} of \IT{Id\SB{A}(a,b)} is means that \IT{p} is a proof that establishes that \IT{a} and \IT{b}, both elements  of the type \IT{A}, are propositionally equal \cite{hott}. The homotopical interpretation of the identity type is responsible for the rise of Homotopy Type Theory. According with this interpretation, the witness \IT{p} can be seen as a homotopical path between the points \IT{a} and \IT{b} within a topological space \IT{A}. This simple interpretation is responsible for groundbreaking results. Most of these were explored in \cite{hott, Steve1}. Inspired by this interpretation, the concept of sequences of rewrites or computational paths was proposed in \cite{Ruy1}. A sequence of rewrites is a computational path \IT{s} between \IT{a} and \IT{b}.


\subsection{Equality Axioms}

To explain the exact meaning of a computational path, we need to look at the equality axioms for Type Theory. Based on the equality theory for the lambda calculus together with $\beta$-reductions and $\eta$-reductions, the so-called $\lambda$$\beta$$\eta$-equality, the following axioms are formalized in type theory \cite{Ruy1}:

\bigskip

\noindent
\begin{bprooftree}
\hskip -0.3pt
\alwaysNoLine
\AxiomC{\IT{N : A}}
\AxiomC{\IT{[x : A]}}
\UnaryInfC{\IT{M : B}}
\alwaysSingleLine
\LeftLabel{(\IT{$\beta$}) \quad}
\BinaryInfC{\IT{($\lambda$x.M)N = M[N/x] : B}}
\end{bprooftree}
\begin{bprooftree}
\hskip 11pt
\alwaysNoLine
\AxiomC{\IT{[x : A]}}
\UnaryInfC{\IT{M = M' : B}}
\alwaysSingleLine
\LeftLabel{\IT{($\xi$)} \quad}
\UnaryInfC{\IT{$\lambda$x.M = $\lambda$x.M' : ($\Pi$x : A)B}}
\end{bprooftree}

\bigskip

\noindent
\begin{bprooftree}
\hskip -0.5pt
\AxiomC{\IT{M : A}}
\LeftLabel{\IT{($\rho$)} \quad}
\UnaryInfC{\IT{ M = M : A}}
\end{bprooftree}
\begin{bprooftree}
\hskip 58pt
\AxiomC{\IT{ M = M' : A}}
\AxiomC{\IT{N : ($\Pi$x : A)B}}
\LeftLabel{\IT{($\mu$)} \quad}
\BinaryInfC{\IT{NM = NM' : B}}
\end{bprooftree}

\bigskip

\noindent
\begin{bprooftree}
\hskip -0.5pt
\AxiomC{\IT{M = N : A}}
\LeftLabel{\IT{($\sigma$)} \quad}
\UnaryInfC{\IT{N = M : A}}
\end{bprooftree}
\begin{bprooftree}
\hskip 63pt
\AxiomC{\IT{N : A}}
\AxiomC{\IT{M = M' : ($\Pi$x : A)B}}
\LeftLabel{\IT{($\upsilon$)} \quad}
\BinaryInfC{\IT{ MN = M'N : B}}
\end{bprooftree}

\bigskip

\noindent
\begin{bprooftree}
\hskip -0.5pt
\AxiomC{\IT{M = N : A}}
\AxiomC{\IT{ N = P : A}}
\LeftLabel{\IT{($\tau$)} \quad}
\BinaryInfC{\IT{M = P : A}}
\end{bprooftree}

\bigskip

\noindent
\begin{bprooftree}
\hskip -0.5pt
\AxiomC{\IT{M : ($\Pi$x : A)B}}
\LeftLabel{\IT{($\eta$)} \quad}
\RightLabel {\IT{(x $\notin$ FV(M))}}
\UnaryInfC{\IT{($\lambda$x.Mx) = M : ($\Pi$x : A)B}}
\end{bprooftree}

\bigskip

\noindent One can check the original $\lambda$$\beta$$\eta$-equality axioms and rules of inference in \cite{lambda}. An important operation not mentioned in these axioms is the change of bound variables. This operation is denoted by $\alpha$ in the literature \cite{lambda}.

\begin{mydef}[Computational path \cite{Ruy1}] 
\normalfont Let  \IT{a} and \IT{b} be elements of a type \IT{A}. A computational path from \IT{a} to \IT{b} is a sequence of rewrites and substitutions, i.e., a sequence of definitional equalities, that starts from \IT{a} and arrives at \IT{b}.
\end{mydef}

\begin{notation}

\noindent \normalfont We use the notation \IT{a =\SB{t} b : A} to indicate that \IT{t} is a computational path between the elements \IT{a} and \IT{b} of type \IT{A}.
\end{notation}

\noindent The definition of a computational path is the reason why we needed to introduce these equality axioms. A rewrite can be identified by the equality axiom (or $\alpha$ operation) that originated it. Thus, a computational path (or sequence of rewrites) can be understood as a sequence of equality identifiers. A path \IT{t} is sometimes called the \IT{rewrite reason}, since it is the reason that justifies the equality of elements \IT{a} and \IT{b} in  \IT{a =\SB{t} b : A} \cite{Ruy1}.
\par Using this new entity, the sequence of rewrites, it is possible to give a new formalization to the identity type. From \IT{a =\SB{t} b : A}, we derive the element \IT{t(a,b)} (indicating that \IT{t} is a sequence of rewrites between \IT{a} and \IT{b}, both elements of the type \IT{A}). Furthermore, this element \IT{t(a,b)} is of the type \IT{Id\SB{A}(a,b)}, \IT{t(a,b)} : \IT{Id\SB{A}(a,b)}. That way, we conclude the following introduction rule for the identity type:

\begin{prooftree}
\AxiomC{\IT{a =\SB{t} b : A}}
\UnaryInfC{\IT{t(a,b) : Id\SB{A}(a,b)}}
\end{prooftree}

\noindent Using computational paths, we can make a complete formalization of the identity type, including rules such as the \IT{Id-elimination} and \IT{Id-induction}. One can check this formalization, together with the advantages in terms of simplicity, in the thorough formulation given in \cite{Ruy1}.

\subsection{Equality Properties}

Three equality rules will be essential to the categorical interpretation \cite{Ruy1}:

\begin{prooftree}
\hskip -180pt
\AxiomC{\IT{a =\SB{t} b : A}}
\AxiomC{\IT{b =\SB{u} c : A}}
\LeftLabel{\IT{transitivity}}
\BinaryInfC{ \IT{a =\SB{$\tau$(t,u)} c : A}}
\end{prooftree}

\begin{prooftree}
\hskip -248pt
\AxiomC{\IT{a : A}}
\LeftLabel{\IT{reflexivity}}
\UnaryInfC{\IT{a =\SB{$\rho$} a : A}}
\end{prooftree}
\begin{prooftree}
\hskip -240pt
\AxiomC{\IT{a =\SB{t} b : A}}
\LeftLabel{\IT{symmetry}}
\UnaryInfC{\IT{b =\SB{$\sigma$(t)} a : A}}
\end{prooftree}




\subsection{Term Rewriting System}

This topic is essential to understand the upcoming categorical aspects of sequences of rewrites. Let's start with some examples:

\begin{example}
\noindent \normalfont Consider the path  \IT{a =\SB{t} b : A}. By the symmetric property, we obtain  \IT{b =\SB{$\sigma$(t)} a : A}. What if we apply the property again on the path \IT{$\sigma$(t)}? We would obtain a path  \IT{a =\SB{$\sigma$($\sigma$(t))} b : A}. Since we applied the inversion twice in sucession, we obtained a path that is equivalent to the initial path \IT{t}. For that reason, we conclude the act of applying the symmetry twice in succession is a redundancy. We say that the path \IT{$\sigma$($\sigma$(t))} can be reduced to the path \IT{t}.
\end{example}

\begin{example}
\noindent \normalfont  Consider the reflexive path \IT{a =\SB{$\rho$} a : A}. Since it is a trivial path, the symmetric path \IT{a =\SB{$\sigma$($\rho$)} a : A} is equivalent to the initial one. For that reason, the application of the symmetry on the reflexive path is a redundancy. The path \IT{$\sigma$($\rho$)} can be reduced to the path \IT{$\sigma$}.

\end{example}
\begin{example}
\noindent \normalfont Consider the path   \IT{a =\SB{t} b : A} and the inverse path \IT{b =\SB{$\sigma$(t)} a : A}. We can apply the transitive property in these paths, obtaining \IT{a =\SB{$\tau$(t,$\sigma$(t))} a : A}. Since the paths are inversions of each other, the transitive path is equivalent to the trivial path \IT{$\rho$}. Therefore, this transitive application is a redundancy. The path \IT{$\tau$(t,$\sigma$(t))} can be reduced to the trivial path \IT{$\rho$}.
\end{example}

\noindent The redundancies found in these examples were pretty simple and straightforward. Nonetheless, some redundancies can be quite complex. In this work, we will only care about the transitivity, reflexivity and symmetry ones. There exist a lot more redundancies derived from the other equality axioms though. The system with the rules that resolve all these redundancies is called \IT{LND\SB{EQ}-TRS}. In the sequel we will take the time to show the rules of this system which we will need to use to obtain the categorical interpretation. Nevertheless, one can check the full system together with a thorough explanation of each rule in \cite{Ruy1}.

\begin{mydef}[\IT{rw}-rule] 
\normalfont A \IT{rw}-rule is any of the rules defined in \IT{LND\SB{EQ}-TRS}. 
\end{mydef}

\begin{mydef}[\IT{rw}-contraction]
\normalfont  Let \IT{a} and \IT{b} be computational paths. We say that  \IT{a $\rhd$\SB{1rw} b} (read as: \IT{a} \IT{rw}-contracts to \IT{b})  iff we can obtain \IT{b} from \IT{a} by an application of only one \IT{rw}-rule.
\end{mydef} 

\begin{mydef}[\IT{rw}-reduction] 
\normalfont Let \IT{a} and \IT{b} be computational paths. We say that \IT{a $\rhd$\SB{rw} b} (read as: \IT{a} \IT{rw}-reduces to \IT{b})  iff we can obtain \IT{b} from \IT{a} by a finite number of applications of \IT{rw}-contractions.
\end{mydef}

\begin{mydef}[\IT{rw}-equality]
\normalfont  Let \IT{a} and \IT{b} be computational paths.  We say that \IT{a =\SB{rw} b} (read as: \IT{a} is \IT{rw}-equal to \IT{b}) iff \IT{b} can be obtained from \IT{a} by a finite (perhaps empty) series of \IT{rw}-contractions and reversed \IT{rw}-contractions. In other words, \IT{a =\SB{rw} b} iff there exists a sequence \IT{ R\SB{0},....,R\SB{n}}, with $n \geq 0$, such that

\centering \IT{($\forall$i $\leq$ n - 1) (R\SB{i}$\rhd$\SB{1rw} R\SB{i+1} or R\SB{i+1} $\rhd$\SB{1rw} R\SB{i})}

\centering \IT{ R\SB{0} $\equiv$ \IT{a}, \quad R\SB{n} $\equiv$ \IT{b}}
\end{mydef}

\begin{prop} \normalfont \IT{rw}-equality is transitive, symmetric and reflexive.
\end{prop}

\begin{proof}
The reflexivity is straightfoward. For any sequence of rewrites \IT{x}, just take the trivial sequence \IT{x =\SB{rw} x}. For transitivity, for sequences  \IT{x =\SB{rw} y} and \IT{y =\SB{rw} z}, just concatenate the two to obtain the transitive sequence \IT{x =\SB{rw} z}. Since in a \IT{rw}-equality sequence every step can be an inverse \IT{rw}-contraction, it is natural that if \IT{x =\SB{rw} y}, then \IT{y =\SB{rw} x}. To understand this fact, just see that a contraction in the sequence  \IT{x =\SB{rw} y} will be an inverse contraction in  \IT{y =\SB{rw} x} and an inverse contraction in \IT{x =\SB{rw} y} will be a contraction in  \IT{y =\SB{rw} x}.
\end{proof}

\noindent One important fact is that \IT{LND\SB{EQ}-TRS} has the terminating property, i.e., given any path \IT{s}, we can always find a \IT{t} free of redundancies such that  \IT{s $\rhd$\SB{rw} t}. Another important property is that  \IT{LND\SB{EQ}-TRS} is confluent, i.e., if  \IT{s $\rhd$\SB{rw} n} and  \IT{s $\rhd$\SB{rw} m}, \IT{n $\neq$ m}, there always exists a \IT{t} such that  \IT{m $\rhd$\SB{rw} t} and  \IT{n $\rhd$\SB{rw} t}. One can check the proof of these properties in \cite{Anjo1,Ruy2,Ruy3}.

\subsection{The Relevant \IT{rw}-Rules}
	
In the last section, we introduced the system \IT{LND\SB{EQ}-TRS}. Giving a detailed explanation of each \IT{rw}-rule is not one of the objectives of this work. Nonetheless, the transitive (\IT{$\tau$}), reflexive (\IT{$\rho$}) and symmetric (\IT{$\sigma$}) rules are essential to our future results. In this section, we show all the rules that will be used in this work:

\begin{itemize}
\item Rules involving $\sigma$ and $\rho$




\IT{$\sigma$($\rho$) $\rhd$\SB{sr} $\rho$   \quad $\sigma$($\sigma$(r)) $\rhd$\SB{ss} r }

\item Rules involving $\tau$






\IT{$\tau$(r,$\sigma$(r)) $\rhd$\SB{tr} $\rho$} \quad \IT{$\tau$($\sigma$(r),r) $\rhd$\SB{tsr} $\rho$} \quad \IT{$\tau$(r,$\rho$) $\rhd$\SB{trr} r} \quad \IT{$\tau$($\rho$,r) $\rhd$\SB{tlr} r}


\item Rule involving $\tau$ and $\tau$




\IT{$\tau$($\tau$(t,r),s) $\rhd$\SB{tt} $\tau$(t, $\tau$(r,s))}

\end{itemize}

\noindent One can check the natural deductions that originate each rule in \IT{appendix A}. 


\section{Categorical Interpretation}

As already mentioned, the main result of this work is the interpretation of sequences of rewrites using the categorical point-of-view. Here comes the main points of such an interpretation.

\subsection{Basic Concepts}

To fully understand the interpretation, some fundamental concepts need to be defined. These are: categories, isomorphism between arrows, and groupoids.

\begin{mydef} [Category \cite{Steve2}]
\normalfont A category is the structure with the following elements and rules:

\begin{itemize}
\item Objects: Generally represented by capital letters (e. g. \IT{A, B)}. The set of all objects is called \IT{C\SB{0}}.
\item Arrows: Generally represented by function letters (e. g. \IT{f, g}). Arrows are morphisms between the objects. The set of all arrows is called  \IT{C\SB{1}}.
\item Domain and codomain: Given an arrow \IT{f : A$\rightarrow$  B}, the domain of \IT{f} is \IT{dom(f) = A} and the codomain is \IT{cod(f) = B}.
\item Compositions: Given arrows  \IT{f : A$\rightarrow$  B} and  \IT{g : B$\rightarrow$  C}, there always exists an arrow \IT{g $\circ$ f : A $\rightarrow$ C}. Such arrow is called the composite of \IT{f} and \IT{g}.
\item  Identity arrow: For each object \IT{A} there exists an arrow \IT{1\SB{A}: A $\rightarrow$ A} called the identity arrow.
\item Rule of Associativity: Given arrows \IT{f : A$\rightarrow$  B}, \IT{g : B$\rightarrow$  C} and \IT{h : C$\rightarrow$  D}, we always have that \IT{h $\circ$ (g $\circ$ f) = (h $\circ$ g) $\circ$ f}.
\item Unit Rule: For all arrow \IT{f : A$\rightarrow$  B}, \IT{f $\circ$ 1\SB{A} = f = 1\SB{B} $\circ$  f}.
\end{itemize}

\end{mydef}

\noindent Any structure that has this data and obeys these rules is considered a category. One simple example is the structure that the objects are sets and the arrows are functions. This category is known as \BF{Sets}.

\begin{mydef}[Arrow isomorphism \cite{Steve2}] 
\normalfont Let \IT{f : A$\rightarrow$  B} be an arrow of any category \BF{C}. \IT{f} is called an isomorphism if there exists a \IT{g: B$\rightarrow$  A $\in$ \IT{C\SB{1}}} such that \IT{g $\circ$ f = 1\SB{A}} and  \IT{f $\circ$ g = 1\SB{B}}. \IT{g} is called the inverse of \IT{f} and can be written as \IT{f\SP{-1}}.
\end{mydef}

\noindent Note that this definition generalizes the concept of isomorphism in Mathematics. For the category of \BF{Sets}, an isomorphism is just a bijective function, since every bijective function has an inverse.

\begin{mydef}[Categorical Groupoid \cite{group}]
\normalfont A groupoid is a category such that every arrow is an isomorphism.
\end{mydef}

\noindent We took the care to indicate that we are giving the categorical definition for the concept of groupoid. This care arose due to the fact that groupoids can be defined in a purely algebraic way.  Furthermore, the algebraic definition was already explored using type theory and the identity type, as can be seen in \cite{Streicher}.

\subsection{The Interpretation}

With the definitions proposed in the last section, we can give the categorical interpretation.

\begin{mydef}[\IT{A\SB{rw}}] \normalfont Let \IT{A\SB{rw}} be the structure composed by objects and arrows with the following interpretation:

\begin{itemize}

\item Objects: The objects of \IT{A\SB{rw}} are elements \IT{a} of the type \IT{A}, i.e. \IT{a : A}.
\item Arrows: An arrow  \IT{s : a$\rightarrow$  b} between elements \IT{a,b} of the type \IT{A} is a sequence of rewrite between these elements. In other words, there exists an arrow \IT{s} between elements \IT{a} and \IT{b} iff there exists a computational path \IT{s} such that \IT{a =\SB{s} b}.
\end{itemize}

\end{mydef}

\begin{prop} \normalfont \IT{A\SB{rw}} has a weak categorical structure.
\end{prop}

\begin{proof}
We need to define composition, the identity arrow and show that these elements follow the associative and unit rules (in a weak sense), respectively:

\begin{itemize}
\item Compositions: Given arrows (paths)  \IT{s : a$\rightarrow$  b} and  \IT{r : b$\rightarrow$  c}, we need to find an arrow \IT{t: a$\rightarrow$ c} such that \IT{t = r $\circ$ s}. To do that, we first need to define the meaning of a composition of paths. Since a sequence of rewrites has the transitive axiom, it is natural to define the composition as an application of the transitivity, i.e.\ \IT{t = r $\circ$ s = $\tau$(s,r) }. Therefore, for any \IT{s : a$\rightarrow$  b} and  \IT{r : b$\rightarrow$  c}, we always have a \IT{t = r $\circ$ s = $\tau$(s,r) }.

\item Associativity of the composition: Given arrows \IT{s : a$\rightarrow$  b}, \IT{r : b$\rightarrow$  c} and \IT{t : c$\rightarrow$  d}, we need to conclude that \IT{t $\circ$ (r $\circ$ s) = (t $\circ$ r) $\circ$ s}. Substituting the compositions by the correspoding transitivities, we write:

\begin{center}
\IT{$\tau$($\tau$(s,r),t) = $\tau$(s,$\tau$(r,t))}
\end{center}

 \noindent This equality is clearly justified by the  \IT{tt rw}-rule (in a weak sense that will  be explained soon):

\begin{center}
\centering \IT{$\tau$($\tau$(s,r),t)  $\rhd$\SB{tt}  $\tau$(s,$\tau$(r,t))}
\end{center}

\item Identity: For any object \IT{a}, consider the reflexive path \IT{a =\SB{$\rho$} a} as the identity arrow \IT{1\SB{a}}. Let's call this path \IT{$\rho$\SB{a}} (to indicate that it is a reflexive path of \IT{a}).

\item Unit rule: For any arrow \IT{s: a$\rightarrow$ b}, we need to show that \IT{s $\circ$ 1\SB{a} = s = 1\SB{b} $\circ$ s}. This can be showed (also in a weak sense) with a straightforward application of \IT{rw}-rules:
\begin{center}
\IT{s $\circ$ 1\SB{a} = s $\circ$ $\rho$\SB{a} = $\tau$($\rho$\SB{a},s)  $\rhd$\SB{tlr} s}

\IT{1\SB{b} $\circ$ s = $\rho$\SB{b} $\circ$ s  = $\tau$(s,$\rho$\SB{b})  $\rhd$\SB{trr} s}
\end{center}

\end{itemize}

\noindent An interesting question is whether \IT{rw}-reduction really implies equality in the general sense. For example, does \IT{$\tau$($\tau$(s,r),t)  $\rhd$\SB{tt}  $\tau$(s,$\tau$(r,t))} really implies that \IT{$\tau$($\tau$(s,r),t) = $\tau$(s,$\tau$(r,t))}? The answer is \BF{no}. In fact, the reduction implies that the two sides are equal up to \BF{\IT{rw}-equality}. It means that these equations do not hold "on the nose", i.e., in the usual sense of equality, but hold if we consider equality in the sense of \IT{rw}-equality. For that reason, since the equality does not hold in the more general sense, we call this a \BF{weak} categorical structure.
\end{proof}

\noindent In fact,  to hold up to \IT{rw}-equality is the same as saying that it holds up to \BF{propositional equality}. The reason for that, is that every \IT{rw}-rule creates an element (i.e., a proof) for a specific identity type. An example is the   \IT{$\rhd$\SB{tt}}-rule, which creates a proof for \IT{Id\SB{Id\SB{A}(x,z)}($\tau$($\tau$(t,r),s),$\tau$(t,$\tau$(r,s)))}. One can check \cite{Ruy1} for more details.

One interesting fact is that these weak categorical structures are very common in mathematics. In fact, the homotopical interpretation of the identity type as a topological path generates a weak structure too, since paths are associative only up to homotopy.

\begin{prop} \normalfont \IT{A\SB{rw}} has a weak groupoidal structure.
\end{prop}

\begin{proof}

We need to show that every arrow is an isomorphism. To do that, for every arrow \IT{s: a$\rightarrow$ b} we need to show a \IT{t: b$\rightarrow$ a} such that \IT{t $\circ$ s = 1\SB{a}} and  \IT{s $\circ$ t = 1\SB{b}}. To do that, recall that every rewrite \IT{s} has an inverse rewrite \IT{$\sigma$(s)}. Put \IT{t = $\sigma$(s)}. Thus:
\begin{center}

\IT{s $\circ$ t = s $\circ$ $\sigma$(s) = $\tau$($\sigma$(s),s) $\rhd$\SB{tsr} $\rho$\SB{b}}

\IT{t $\circ$ s = $\sigma$(s) $\circ$ s  = $\tau$(s,$\sigma$(s)) $\rhd$\SB{tr} $\rho$\SB{a}}
\end{center}

\noindent This is the same situation discussed in the previous proposition. The equalities only hold up to \IT{rw}-equality. For that reason, we cannot say that they are an isomorphism in the strict sense. Instead, we call this an equivalence. We say that \IT{A\SB{rw}} has a weak groupoidal structure.

\end{proof}

\noindent With these results, we conclude that a type \IT{A}, which is represented by the category  \IT{A\SB{rw}}, has a weak groupoidal structure. This result gives a mathematical meaning to sequences of rewrites that is in accordance with the known results obtained using the usual definition of identity type. Nevertheless, we want to go further and make a brief analysis of a possible higher structure 
for the sequence of rewrites.  Since we have rewrites between elements of the same type, one could think of rewrites between rewrites. In fact, we already have these, the \IT{rw}-rules! This way, we can consider a structure in which the objects are rewrites of rewrites and the arrows are the morphisms between them:

\begin{mydef}[\IT{A\SB{2rw}}] \normalfont \IT{A\SB{2rw}(x,y)} is the following structure:

\begin{itemize}
\item Objects: The objects are the morphisms of \IT{A\SB{rw}} between objects \IT{x} and \IT{y}, i.e., the sequence of rewrites starting at \IT{x} and ending at \IT{y}.
\item Arrows: A rewrite of a sequence of rewrites, as we already know, is just the application of \IT{rw}-rules. That way, there exists an arrow \IT{$\theta$ : s $\rightarrow$ t} iff \IT{s =\SB{rw} t}, where s and t are paths of \IT{A\SB{rw}}.
\end{itemize}

\noindent It is important to note that every pair of objects \IT{x} and \IT{y} of \IT{A\SB{rw}} generates a structure \IT{A\SB{2rw}(x,y)}.

\end{mydef}

\begin{prop} \normalfont \IT{A\SB{2rw}(x,y)} has a weak categorical structure.
\end{prop}

\begin{proof}

Remember that in \BF{proposition 2.1} we already proved that \IT{rw}-equality is transitive, symmetric and reflexive. We also know that the usual equality has these properties. We discussed that some redundancies arise from these properties and are resolved by rules that we called \IT{rw}-rules. Thus, we can say that a \IT{rw}-rule is a rule that removes redundancies generated by the properties of equality. All the \IT{rw}-rules used in this work were derived from the properties of transitivity, symmetry and reflexivity of the equality. Since \IT{rw}-equality also has these properties, we can use an analogous process to obtain redundancies involving \IT{rw}-equality. Moreover, we can obtain analogous rules that remove these redundancies (just change the equality for \IT{rw}-equality in the natural deductions of \IT{appendix A}). These rules are called \IT{rw\SB{2}}-rules. (Notation: \IT{$\rhd$\SB{tt\SB{2}}, $\rhd$\SB{tr\SB{2}}, $\rhd$\SB{tlr\SB{2}}}, etc). Since we are working with \IT{rw\SB{2}}-rules, we can think of contractions, reductions and, of course, \IT{rw\SB{2}}-equality. These definitions are analogous to the usual definitions used for the \IT{rw}-rules. With that, we have the following relations (Suppose \IT{$\theta$, $\alpha$} and \IT{$\psi$} as arrows of  \IT{A\SB{2rw}(x,y)} such that the composition \IT{$\psi$ $\circ$ ($\alpha$ $\circ$ $\theta$)} makes sense and \IT{$\theta$: s $\rightarrow$ t}):

\begin{center}

\IT{$\tau$($\tau$($\theta$,$\alpha$),$\psi$)  $\rhd$\SB{tt\SB{2}}  $\tau$($\theta$,$\tau$($\alpha$,$\psi$))} 

\IT{$\tau$($\rho$\SB{s},$\theta$) $\rhd$\SB{tlr\SB{2}} $\theta$} \quad \IT{$\tau$($\theta$,$\rho$\SB{t}) $\rhd$\SB{trr\SB{2}} $\theta$} 

\end{center}

\noindent This case is similar to what we discussed in \IT{A\SB{rw}}, in the sense that the categorical equalities do not hold in the general sense. They hold up to \BF{\IT{rw\SB{2}}-equality}. For this reason, we conclude that this is a weak categorical structure.

\end{proof}

\begin{prop} \normalfont \IT{A\SB{2rw}(x,y)} has a weak groupoidal structure.
\end{prop}

\begin{proof}
This proof is analogous to \BF{proposition 3.2}. We just need to change the \IT{rw}-rule for the corresponding \IT{rw\SB{2}}-rule. That way, the equality will hold up to \IT{rw\SB{2}}-equality and the arrows will be weak isomorphisms, i.e., equivalences.
\end{proof}
 
\noindent What if we want to go further? Instead of a rewrite of rewrites, we could think of a rewrite of rewrite of rewrites (or even more). Since a morphism of a rewrite of rewrites is a \IT{rw}-equality, we could think of a morphism between a rewrite of rewrite of rewrites as being a \IT{rw\SB{2}}-equality. In fact, we can do this process for any \IT{n}. For this, we just need to define a \IT{rw\SB{n}}-rule:

\begin{mydef}[\IT{rw\SB{n}}-rule] \normalfont For \IT{n $\geq$ 2}, a \IT{rw\SB{n}}-rule is a rule that removes a redundancy generated by \IT{rw\SB{n-1}}-equality. For \IT{n = 1}, just use the definition of \IT{rw}-rule.
\end{mydef}

\begin{prop} \normalfont \IT{rw\SB{n}}-equality is transitive, symmetric and reflexive.
\end{prop}

\begin{proof}
Analogous to \BF{proposition 2.1}. Too see that, just remember that the definition \IT{rw\SB{n}}-equality is similar to \IT{rw}-equality, but using \IT{rw\SB{n}}-contractions (or inverse contractions) in each step.
\end{proof}

\noindent One of our objectives is to generalize the structure  \IT{A\SB{nrw}}. Remember that every pair of objects \IT{x} and \IT{y} of \IT{A\SB{rw}} generates a structure  \IT{A\SB{2rw}(x,y)}. If we take a pair of objects \IT{$\theta$} and \IT{$\alpha$} from \IT{A\SB{2rw}(x,y)}, we can generate a structure \IT{A\SB{3rw}($\theta$,$\alpha$)\SB{(x,y)}}. If we continue this reasoning, we can think of a general structure  \IT{A\SB{nrw}}. But before we do that, let us simplify the notation. Looking at \IT{A\SB{3rw}($\theta$,$\alpha$)\SB{(x,y)}}, we can note that the notation starts to get complicated, due to the fact the structure at each level is generated by a pair of objects of the previous one. The results that we want to achieve is independent of the choice of these objects. For that reason, we will just write \IT{A\SB{nrw}}. Nevertheless, one has to have in mind the fact that, when we write \IT{A\SB{nrw}}, we are talking about a structure generated by a pair of objects of \IT{A\SB{(n-1)rw}}, which, in its turn, is generated by a pair of objects of \IT{A\SB{(n-2)rw}}, etc. In fact, there exists a mathematical structure that captures that idea, known as globular set. Nonetheless, we will not formally define this structure in this work, since it is not needed to achieve our results. One can check the formal definition of globular set (and its many applications) in \cite{Tom}.

\begin{mydef}[\IT{A\SB{n}\SB{rw}}] \normalfont for \IT{n $\geq$ 2}, \IT{A\SB{nrw}} is the following structure:

\begin{itemize}

\item Objects: The objects are the morphisms between a pair of objects of \IT{A\SB{(n-1)rw}}.
\item Arrows: There exists an arrow \IT{$\theta$ : s $\rightarrow$ t} iff \IT{s =\SB{rw\SB{n-1}} t}, where s and t are objects of \IT{A\SB{nrw}}

\end{itemize}

\noindent For \IT{ n = 1}, just take the usual definition of \IT{A\SB{rw}}.

\end{mydef}

\begin{prop} \normalfont \IT{A\SB{nrw}} has a weak categorical structure
\end{prop}

\begin{proof}
 Using the fact that \IT{rw\SB{n - 1}}-equality (if \IT{n = 1}, consider \IT{rw\SB{0}}-equality as the usual equality)  is transitive, symmetric and reflexive, we have \IT{rw\SB{n}}-rules that resolve transitive, symmetric and reflexive redundancies, analogous to the \IT{$\rhd$\SB{tt}, $\rhd$\SB{tr}, $\rhd$\SB{tsr}} rules. That way, this proof is analogous to the one of \BF{proposition 3.3}. Again, it is important to emphasize that the equalities will hold up to \BF{\IT{rw\SB{n}}-equality}, resulting in a weak structure.
\end{proof}

\begin{prop} \normalfont  \IT{A\SB{nrw}} has a weak groupoidal structure
\end{prop}

\begin{proof}
Analogous to the proof of \BF{proposition 3.4}.
\end{proof}

\noindent With that, we concluded the interpretation of every structure \IT{A\SB{nrw}}. Nevertheless, this work is just the initial part that needed to be done to achieve even more important results involving this interpretation in  the future. One of these results is to formalize all these structures in one big and complex higher dimensional structure. We discussed that if we get a pair of objects of any structure \IT{A\SB{nrw}}, we generate a new structure \IT{A\SB{(n+1)rw}}. As we can see, this process is combinatorial. We also proved that, if we look separately at each structure, the groupoidal and categorical equalities hold up to the corresponding \IT{rw}-equality. If we look closely, this \IT{rw}-equality is, in fact, the \BF{morphism} of the next level! If we take all those \IT{A\SB{nrw}} (all up to infinity), we could try to glue them together in a unique multidimensional combinatorial structure, where at each level, the equalities hold up to the morphism of the next one. In fact, a structure with these properties already exists in mathematics. It is known as \BF{weak \IT{$\omega$}-groupoid}. Gluing all these structures to build this higher dimensional one is a hard task though. The reason for that is the fact that a weak \IT{$\omega$}-groupoid has higher dimensional compositions that are hard to define and these compositions have additional laws that need to be checked. In addition, the structure needs to obey other laws that are also hard to define and check, called \BF{coherence laws}. For that reason, an attempt to do that is outside the scope of this work. Nonetheless, one of the results of this work is that it can be understood as the initial step towards the construction of this higher structure, using sequences of rewrites as the fundamental concept.


\section{Conclusion}

This work was motivated by the addition of a new entity in type theory, the sequence of rewrites \cite{Ruy5}. Although a sequence of rewrites has a clear computational definition, it was lacking a mathematical interpretation. Believing that the addition of this new entity greatly reduces the complexity of type theory, we proposed a mathematical formalization for a sequence of rewrites. With that, our objective was to show that a sequence of rewrites is not just a computational entity without any mathematical meaning. In fact, we showed that it has a formal interpretation in a well-established theory of mathematics, i.e.\ Category Theory. Therefore, we achieved our objective of developing a denotational semantics for this new concept.

We have obtained some results involving the categorical structure that was derived from our interpretation. We concluded that our interpretation is in accordance with some known results, such as types have a groupoidal structure. We have also interpreted the meaning of a rewrite of a sequence of rewrites. Motivated by that, we went further and examined complicated structures involving rewrites of objects that are also rewrites of other objects. We have generalized this structure, showing that \IT{A\SB{nrw}} has a weak groupoidal structure. In the end, we made a connection between these structures with the so-called weak \IT{$\omega$}-groupoid. This has opened the way, in a future work, for an alternative proof of the relevant fact that types are weak \IT{$\omega$}-groupoids \cite{Benno}, using the interpretation given here.

\newpage


\newpage
\bibliographystyle{plain}
\bibliography{Biblio}


\newpage

\appendix
\section{Natural Deduction Derivations} \label{App:AppendixA}

In this section, we show natural deductions for each \IT{rw}-rule used in this work (i.e., the rules presented in \BF{section 2.4}). All deductions have been taken from \cite{Ruy1}:

\begin{itemize}
\item Rules involving $\sigma$ and $\rho$

\begin{prooftree}
\AxiomC{\IT{x =\SB{$\rho$} x : A}}
\RightLabel{\quad \IT{$\rhd$\SB{sr} \quad x =\SB{$\rho$} x : A}}
\UnaryInfC{\IT{x =\SB{$\sigma$($\rho$)} x : A}}
\end{prooftree}

\begin{prooftree}
\AxiomC{\IT{x =\SB{r} y : A}}
\UnaryInfC{\IT{y =\SB{$\sigma$(r)} x : A}}
\RightLabel{\quad \IT{$\rhd$\SB{ss} \quad x =\SB{r} y : A}}
\UnaryInfC{\IT{x =\SB{$\sigma$($\sigma$(r))} y : A}}
\end{prooftree}

From these deductions, we obtain the following rules:

\IT{$\sigma$($\rho$) $\rhd$\SB{sr} $\rho$ }

\IT{$\sigma$($\sigma$(r)) $\rhd$\SB{ss} r }
\bigskip

\item Rules involving $\tau$

\begin{prooftree}
\AxiomC{\IT{x =\SB{r} y : A}}
\AxiomC{\IT{y =\SB{$\sigma$(r)} x : A}}
\RightLabel{\quad \IT{$\rhd$\SB{tr} \quad x =\SB{$\rho$} x : A}}
\BinaryInfC{\IT{x =\SB{$\tau$(r,$\sigma$(r))} x : A}}
\end{prooftree}

\begin{prooftree}
\AxiomC{\IT{y =\SB{$\sigma$(r)} x : A}}
\AxiomC{\IT{x =\SB{r} y : A}}
\RightLabel{\quad \IT{$\rhd$\SB{tsr} \quad y =\SB{$\rho$} y : A}}
\BinaryInfC{\IT{y =\SB{$\tau$($\sigma$(r),r)} y : A}}
\end{prooftree}

\begin{prooftree}
\AxiomC{\IT{x =\SB{r} y : A}}
\AxiomC{\IT{y =\SB{$\rho$} y : A}}
\RightLabel{\quad \IT{$\rhd$\SB{trr} \quad x =\SB{r} y : A}}
\BinaryInfC{\IT{x =\SB{$\tau$(r,$\rho$)} y : A}}
\end{prooftree}

\begin{prooftree}
\AxiomC{\IT{x =\SB{$\rho$} x : A}}
\AxiomC{\IT{x =\SB{r} y : A}}
\RightLabel{\quad \IT{$\rhd$\SB{tlr} \quad x =\SB{r} y : A}}
\BinaryInfC{\IT{x =\SB{$\tau$($\rho$),r} y : A}}
\end{prooftree}

Obtained Rules:

\IT{$\tau$(r,$\sigma$(r)) $\rhd$\SB{tr} $\rho$}

\IT{$\tau$($\sigma$(r),r) $\rhd$\SB{tsr} $\rho$}

\IT{$\tau$(r,$\rho$) $\rhd$\SB{trr} r}

\IT{$\tau$($\rho$,r) $\rhd$\SB{tlr} r}

\bigskip

\item Rule involving $\tau$ and $\tau$

\begin{prooftree}
\hskip - 155pt
\AxiomC{\IT{x =\SB{t} y : A}}
\AxiomC{\IT{y =\SB{r} w : A}}
\BinaryInfC{\IT{x =\SB{$\tau$(t,r)} w : A}}
\AxiomC{\IT{w =\SB{s} z : A}}
\BinaryInfC{\IT{ x =\SB{$\tau$(($\tau$,r),s)} z : A}}
\end{prooftree}

\begin{prooftree}
\hskip 4cm
\AxiomC{\IT{x =\SB{t} y : A}}
\AxiomC{\IT{y=\SB{r} w : A}}
\AxiomC{\IT{w=\SB{s} z : A}}
\BinaryInfC{\IT{y =\SB{$\tau$(r,s)} z : A}}
\LeftLabel{\IT{$\rhd$\SB{tt}}}
\BinaryInfC{\IT{x =\SB{$\tau$(t,$\tau$(r,s))} z : A}}
\end{prooftree}

Obtained rule:

\IT{$\tau$($\tau$(t,r),s) $\rhd$\SB{tt} $\tau$(t, $\tau$(r,s))}

\end{itemize}

\end{document}